\documentclass{article}
\usepackage{graphicx} 
\usepackage{amsfonts,amsmath,amsthm}
\usepackage{enumerate}
\usepackage{comment}
\usepackage{amssymb}

\usepackage{xcolor}
\definecolor{darkred}{rgb}{0.45,0,0}
\definecolor{darkblue}{rgb}{0.2, 0.2, 0.9}
\definecolor{darkorange}{rgb}{0.8, 0.4, 0.1}

\definecolor{shadecolor}{rgb}{1,0.8,0.3}
\definecolor{myurlcolor}{rgb}{0.5,0,0}
\definecolor{mycitecolor}{rgb}{0,0,0.8}
\definecolor{myrefcolor}{rgb}{0,0,0.8}
\definecolor{hyperrefcolor}{rgb}{0.5,0,0}

\usepackage[
    colorlinks, 
    citecolor=blue, 
    urlcolor=darkred, 
    final, 
    hyperindex, 
    pagebackref, 
    linkcolor = darkblue
]{hyperref}

\newcommand{\define}[1]{{\bf \boldmath{#1}}\index{#1}}

\theoremstyle{plain}

\newtheorem{thm}{Theorem}

\newtheorem{lem}[thm]{Lemma}

\theoremstyle{remark}

\theoremstyle{definition}

\newcommand{\maps}{\colon}    
\newcommand{\tr}{\operatorname{tr}}
\newcommand{\spann}{\operatorname{span}} 

\newcommand{\R}{{\mathbb R}}  
\newcommand{\C}{{\mathbb C}}  
\newcommand{\Z}{{\mathbb Z}}  
\renewcommand{\H}{{\mathbb H}}  
\renewcommand{\O}{{\mathbb O}}  
\newcommand{\K}{{\mathbb K}}   
\renewcommand{\L}{{\mathbb L}}   

\newcommand{\Aut}{{\rm Aut}}  
\newcommand{\U}{{\rm U}}    
\newcommand{\SU}{{\rm SU}}    
\newcommand{\Spin}{{\rm Spin}}    
\newcommand{\F}{{\rm F}}       
\newcommand{\Stab}{\mathrm{Stab}} 
\renewcommand{\S}{\mathrm{S}} 

\newcommand{\M}{\mathrm{M}}   
\newcommand{\h}{\mathfrak{h}}  

\title{The Standard Model Gauge Group \\
from the Exceptional Jordan Algebra }
\author{John C. Baez\footnote{School of Mathematics, University of Edinburgh, James Clerk Maxwell Building, Peter Guthrie Tait Road, Edinburgh, UK EH9 3FD.} \and
Paul Schwahn\footnote{IMECC, Universidade Estadual de Campinas, Rua S\'ergio Buarque de Holanda 651,
13083-859 Campinas-SP, Brazil}}

\date{\today}

\begin{document}

\maketitle

\begin{abstract}
We construct the Standard Model gauge group using the exceptional Jordan algebra $\h_3(\O)$ and its automorphism group $\F_4$.  The group $\F_4$ acts on pairs of Jordan subalgebras $X \subset B \subset \h_3(\O)$ with $X \cong \h_2(\C)$ and $B \cong \h_3(\C)$, and for any such pair the stabilizer of $X$ intersected with the identity component of the stabilizer of $B$ is isomorphic to the Standard Model gauge group.   Since $\h_2(\C)$ is the Jordan algebra of observables of a qubit and $\h_3(\C)$ is the Jordan algebra of observables of a qutrit, we could say $\h_3(\O)$ is the Jordan algebra of observables of an `octonionic qutrit'.   In this language our result says roughly that the Standard Model gauge group is the group of symmetries of an octonionic qutrit that restrict to act as unitary operators on an ordinary qutrit and, within that, a qubit. 
\end{abstract}

\section{Introduction}

There have been many attempts to relate the Standard Model of particle physics to the octonions $\O$ and the exceptional Jordan algebra $\h_3(\O)$, which consists of self-adjoint $3 \times 3$ matrices of octonions.  For a sample, see \cite{Boyle,Dixon,DV,Furey,GurseyTze,Krasnov} and the many references therein.  A particularly crisp mathematical spinoff of this work is Todorov and Dubois-Violette's \cite{TDV} characterization of the Standard Model gauge group as the intersection of two maximal connected subgroups of $\F_4$, which is the automorphism group of $\h_3(\O)$.  However, this raises the question of how to best understand these maximal connected subgroups.  In attempting to answer this question, we sought to describe them in terms of Jordan subalgebras of $\h_3(\O)$, which led to the two theorems here.

While the gauge group of the Standard Model of particle physics is is often said to be $\U(1) \times \SU(2) \times \SU(3)$, in fact a certain $\Z_6$ subgroup of the center acts trivially on all known particles, so we may alternatively take the gauge group to be the quotient $(\U(1) \times \SU(2) \times \SU(3))/Z_6$, which is isomorphic to
\[  \begin{array}{ccl} 
 \S(\U(2) \times \U(3)) := &
\Big\{ x \in \SU(5) : x = 
\left( 
\begin{array}{c c c c c}
\ast & \ast & 0 & 0 & 0 \\
\ast & \ast & 0 & 0 & 0 \\
0 & 0 & \ast & \ast & \ast \\
0 & 0 & \ast & \ast & \ast \\
0 & 0 & \ast & \ast & \ast 
\end{array}
\right) \; \Big\}.  
\end{array}
\]
This fact is crucial for the $\SU(5)$ and $\Spin(10)$ grand unified theories \cite{Georgi,GeorgiGlashow}; an exposition for mathematicians can be found in \cite{BH}.

Our main result is this characterization of $\S(\U(2) \times \U(3))$ as a subgroup of $\F_4$:

\begin{thm}
\label{thm:1}
Suppose $X,B$ are Jordan subalgebras of $\h_3(\O)$ such that
\[  X \cong \h_2(\C), \;\; B \cong \h_3(\C), \;\; X \subset B . \]
Then 
\[     \Stab(X) \cap \Stab(B)_0 \cong \S(\U(2) \times \U(3)). \]
\end{thm}
\noindent Here $\h_n(\C)$ is the Jordan algebra of $n \times n$ self-adjoint complex matrices.  $\Stab(X)$ is the stabilizer of $X$---that is, the subgroup of $\F_4$ consisting of elements that map $X$ to itself---while $\Stab(B)_0$ is the identity component of the stabilizer of $B$.   

Since $\h_2(\C)$ is the Jordan algebra of observables of a qubit and $\h_3(\C)$ is the Jordan algebra of observables of a qutrit, we could say $\h_3(\O)$ is the Jordan algebra of observables of an `octonionic qutrit'.   In this language, Theorem \ref{thm:1} says roughly that the Standard Model gauge group is the group of symmetries of an octonionic qutrit that restrict to act as unitary operators on an ordinary qutrit and restrict further to act on an ordinary qubit. 

Theorem \ref{thm:1} emerged from another characterization of the Standard Model gauge group, which grew out of the work of Todorov and Dubois-Violette \cite{TDV}:

\begin{thm}
\label{thm:2}
Suppose $A,B$ are Jordan subalgebras of $\h_3(\O)$ such that
\[  A \cong \h_2(\O), \;\; B \cong \h_3(\C), \;\; A \cap B \cong \h_2(\C) . \]
Then 
\[     \Stab(A) \cap \Stab(B)_0 \cong \S(\U(2) \times \U(3)) .\]
\end{thm}

Todorov and Dubois--Violette proved this for a certain standard choice of subalgebras $A$ and $B$.  Thus, the challenge in proving Theorem \ref{thm:2} was to show that every other choice can be mapped to this standard choice using the action of $\F_4$.  This shows that the theorem is not an artifact of a specific choice, but rather a general fact.

We begin in Section \ref{sec:octonions} by constructing the octonion product from $\SU(3)$-invariant operations on $\C$ and $\C^3$.  In Section \ref{sec:standard} we use this to reprove Todorov and Dubois--Violette's special case of Theorem \ref{thm:2}.  We also explain what would happen if we tried to use $\Stab(B)$ rather than $\Stab(B)_0$.  In Section \ref{sec:transitivity} we show that $\F_4$ acts transitively on the set of subalgebras of $\h_3(\O)$ that are isomorphic to $\h_3(\C)$.  We start Section \ref{sec:proofs} with some lemmas leading up to a proof that every Jordan subalgebra of $\h_3(\O)$ isomorphic to $\h_2(\C)$ is contained in a unique Jordan subalgebra isomorphic to $\h_2(\O)$.   This lets us prove that $\F_4$ acts transitively on the set of pairs of Jordan subalgebra $A, B \subset \h_3(\O)$ with $A \cong \h_2(\O)$, $B \cong \h_3(\C)$ and $A \cap B \cong \h_3(\C)$.  Theorem \ref{thm:2} then follows from Todorov and Dubois-Violette's special case.  We conclude by using these results to prove Theorem \ref{thm:1}.  

\section{The octonions from complex vectors}
\label{sec:octonions}

To relate $\h_3(\O)$ to the more familiar Jordan algebra $\h_3(\C)$, it is useful to have a construction of the octonions starting from the complex numbers.  In parallel to the usual construction of the quaternions as $\H =  \R \oplus \R^3$
with multiplication given by 
\[ (a, \mathbf{a})(b ,\mathbf{b}) = (a b - \mathbf{a} \cdot \mathbf{b}, \; a \mathbf{b} + b \mathbf{a} + \mathbf{a} \times \mathbf{b}), \]
we can build the octonions as $\O = \C \oplus \C^3$ with the product
\[ (a, \mathbf{a})(b, \mathbf{b}) = (a b - \langle \mathbf{a}, \mathbf{b}\rangle, \; \overline{a}\, \mathbf{b} + b \mathbf{a} + \mathbf{a} \; \overline{\!\times\!} \;\mathbf{b}) .\]
Here $\langle \mathbf{a}, \mathbf{b}\rangle$ is the usual inner product of vectors in $\C^3$, while the last term is the componentwise complex conjugate of the cross product $\mathbf{a} \times \mathbf{b}$.  

\begin{lem}
\label{lem:octonions_from_complex_vectors}
If we define multiplication on $\C \oplus \C^3$ by
\[  (a , \mathbf{a})(b , \mathbf{b}) = (a b -  \langle \mathbf{a}, \mathbf{b}\rangle, \; a \mathbf{b} + b \mathbf{a} + \mathbf{a} \; \overline{\!\times\!} \; \mathbf{b}) \]
then this space becomes a 8-dimensional normed division algebra over the real numbers, which is therefore isomorphic to the octonions.
\end{lem}

\begin{proof}  It is enough to show that we have a normed division algebra, since the octonions are the only 8-dimensional normed division algebra over the reals.   We use the norm with
\[ \|(a ,\mathbf{a})\|^2 =  |a|^2 + \langle \mathbf{a} , \mathbf{a} \rangle \]
and show that
\[  \|(a, \mathbf{a})(b, \mathbf{b})\| = \|(a, \mathbf{a})\| \|(b , \mathbf{b})\| .\]

We start with
\[ \|(a, \mathbf{a})(b, \mathbf{b})\|^2 = \|(a b - \langle \mathbf{a}, \mathbf{b}\rangle , \; a \mathbf{b} + b \mathbf{a} + \mathbf{a} \; \overline{\!\times\!} \; \mathbf{b}) \|^2 \]
and use the definition of the norm to break this up into two terms:
\[ |a b - \langle \mathbf{a}, \mathbf{b} \rangle|^2 + \| a \mathbf{b} + b \mathbf{a}  + \mathbf{a} \; \overline{\!\times\!} \; \mathbf{b} \|^2 \]
We expand the first term:
\[  |a b - \langle \mathbf{a}, \mathbf{b} \rangle|^2 = |a b|^2 - 2 \mathrm{Re}(a b \langle \mathbf{b} , \mathbf{a} \rangle) + |\langle \mathbf{a}, \mathbf{b}\rangle|^2 \]
and expand the second:
\[ \|a \mathbf{b} + b \mathbf{a} + \mathbf{a} \; \overline{\!\times\!} \; \mathbf{b} \|^2 = \]
\[ \| a \mathbf{b} \|^2 + \|b \mathbf{a}\|^2 + \|\mathbf{a} \; \overline{\!\times\!} \; \mathbf{b}\|^2 + 
 2 \mathrm{Re} \big(a b \langle \mathbf{b}, \mathbf{a} \rangle  + 
\overline{a} \langle \mathbf{b}, \mathbf{a} \; \overline{\!\times\!} \; \mathbf{b} \rangle +
\overline{b} \langle \mathbf{a}, \mathbf{a} \; \overline{\!\times\!} \; \mathbf{b} \rangle  \big) \]
Note that
\[  \langle \mathbf{a}, \mathbf{a} \; \overline{\!\times\!} \; \mathbf{b} \rangle = \overline{\mathbf{a} \cdot (\mathbf{a} \times \mathbf{b})} = 0 \]
by a well-known vector identity which works for complex vectors just as for real ones.  Similarly $\langle \mathbf{b}, \mathbf{a} \; \overline{\!\times\!} \; \mathbf{b} \rangle = 0$.  Thus the expanded second term simplifies to
\[ \| a \mathbf{b} \|^2 + \|\mathbf{a} b\|^2 + \|\mathbf{a} \; \overline{\!\times\!} \; \mathbf{b}\|^2 + 
 2 \mathrm{Re} \big(a b \langle \mathbf{b}, \mathbf{a} \rangle \big) \]
and when we add it to the expanded first term we are left with this:
\[  |a b|^2 + |\langle \mathbf{a}, \mathbf{b}\rangle|^2 + 
 \| a \mathbf{b} \|^2 + \|\mathbf{a} b\|^2 + \|\mathbf{a} \; \overline{\!\times\!} \; \mathbf{b}\|^2 \]
We need to show that this equals
\[ 
 \|(a,\mathbf{a})\|^2 \; \|(b,\mathbf{b})\|^2 = |a|^2 |b|^2 + |a|^2 \|\mathbf{b}\|^2 + |b|^2 \|\mathbf{a}\|^2 + \|\mathbf{a}\|^2 \|\mathbf{b}\|^2 
\]
Three terms match, so it is enough to show
\[ |\langle \mathbf{a}, \mathbf{b}\rangle|^2 +  \|\mathbf{a} \; \overline{\!\times\!} \; \mathbf{b}\|^2 =  \|\mathbf{a}\|^2 \|\mathbf{b}\|^2 \]
This would be a familiar identity if we were working in $\R^3$ with the usual cross product instead of in $\C^3$ with $\; \overline{\!\times\!} \;$.   But $\mathbf{a} \; \overline{\!\times\!} \; \mathbf{b}$ is obtained from $\mathbf{a} \times \mathbf{b}$ by componentwise complex conjugation, so $\|\mathbf{a} \; \overline{\!\times\!} \; \mathbf{b}\|^2 =  \|\mathbf{a} \times \mathbf{b}\|^2$.  Thus, we just need to show
\[ |\langle \mathbf{a}, \mathbf{b}\rangle|^2 +  \|\mathbf{a} \times \mathbf{b}\|^2 =  \|\mathbf{a}\|^2 \|\mathbf{b}\|^2 \]
for vectors in $\C^3$.  Here one can simply write out both sides using components and check that they agree.   \end{proof}

The inner product and conjugated cross product $\; \overline{\! \times \!} \; $ are $\SU(3)$-equivariant operations on $\C^3$, so $\SU(3)$ acts as algebra automorphisms of $\O = \C \oplus \C^3$, and trivially on the subalgebra $\C \subset \O$.  Indeed it is well known that the group of automorphisms of $\O$ preserving a chosen square root of $-1$ is isomorphic to $\SU(3)$, and Yokota's proof of this fact \cite[Thm.\ 1.9.1]{Yokota} supplies the structures used in Lemma \ref{lem:octonions_from_complex_vectors}, though he does not state this result.

\section{The Standard Model gauge group}
\label{sec:standard}

Sitting inside $\h_3(\O)$, there are many choices of Jordan subalgebra $A \cong \h_3(\O)$, $B \cong \h_3(\C)$ such that $A \cap B \cong \h_2(\C)$.  For example, since
\[ \h_3(\O) \; = \;\; \left\{  \left( \begin{array}{ccc}  
                         \alpha  &  z  & y^*    \\  
                         z^*       & \beta & x      \\ 
                         y       & x^*   & \gamma   
\end{array} \right) : \; \alpha,\beta,\gamma \in \R, \; x,y,z \in \O \right\} \qquad \qquad \]
we can take
\begin{equation}
\label{eq:choice_of_AB}
\begin{array}{cclccc} 
 A &=& \left\{  \left( \begin{array}{ccc}  
                         \alpha  &  z  & 0    \\  
                         z^*       & \beta & 0      \\ 
                         0    & 0   & 0   
\end{array} \right) : \; \alpha,\beta \in \R, \; z \in \O \right\} &\cong& \h_2(\O) \\ \\
 B &=& \left\{  \left( \begin{array}{ccc}  
                         \alpha  &  z  & y^*    \\  
                         z^*       & \beta & x      \\ 
                         y       & x^*   & \gamma   
\end{array} \right) : \; \alpha,\beta,\gamma \in \R, \; x,y,z \in \C \right\} &\cong& \h_3(\C) 
\end{array}
\end{equation}
and then
\[  \begin{array}{cclccc}
  A \cap B &=& \left\{  \left( \begin{array}{ccc}  
                         \alpha  &  z  & 0    \\  
                         z^*       & \beta & 0      \\ 
                         0    & 0   & 0   
\end{array} \right) : \; \alpha,\beta \in \R, \; z \in \C \right\} &\cong& \h_2(\C).
\end{array}
\]
Dubois-Violette and Todorov showed that in this example
\[     \Stab(A) \cap \Stab(B)_0 = \S(\U(2) \times \U(3)) .\]
Thus, they showed that Theorem \ref{thm:2} holds for this particular choice of $A$ and $B$.  To prove the theorem in full generality we first need to recall how this special case works.

\begin{lem}
\label{lem:thm:2_special_case}
For the Jordan subalgebras $A$ and $B$ in Equation \eqref{eq:choice_of_AB} we have
\[     \Stab(A) \cap \Stab(B)_0 = \S(\U(2) \times \U(3)). \qedhere \]
\end{lem}

\begin{proof}
The first step is to show that
\[   \Stab(B)_0 \cong (\SU(3) \times \SU(3))/\Z_3 \]
when $B$ is given as in Equation \eqref{eq:choice_of_AB}.
In Proposition \ref{lem:octonions_from_complex_vectors} we constructed the octonions $\O = \C \oplus \C^3$, and we use that description of the octonions here.   In this description, the inner product of octonions is
\[   \langle (a,\mathbf{a}), \; (b, \mathbf{b}) \rangle = \overline{a} b + \langle \mathbf{a}, \mathbf{b} \rangle \]
so the orthogonal complement of $\C \subset \O$ is $\C^3$.
The inner products on $\R$ and $\O$, in turn, give the exceptional Jordan algebra $\h_3(\O) \cong \R^3 \oplus \O^3$ an inner product.  This lets us split the exceptional Jordan algebra as an orthogonal direct sum of subspaces:
\[   \h_3(\O) \cong \h_3(\C) \oplus \h_3(\C)^\perp .\]
Concretely, we have
\[
\begin{array}{ccl}
\h_3(\C)^\perp 
&=& \left\{  \left( \begin{array}{ccc}  
                         0  &  z  & y^*    \\  
                         z^*       & 0 & x      \\ 
                         y       & x^*   & 0   
\end{array} \right) : \;  x,y,z \in \C^3 \subset \O \right\}  \\ \\
&\cong& \{(x,y,z) : \; x,y,z \in \C^3 \}  \\ \\
&\cong& \M_3(\C) 
\end{array}
\]
where $\M_3(\C)$ is the space of $3 \times 3$ complex matrices.  Here we think of $x,y,z \in \C^3$ as column vectors, and assemble them into a $3 \times 3$ matrix.

We thus obtain a vector space isomorphism
\[  \h_3(\O) \cong  \h_3(\C) \oplus \M_3(\C) 
 \]
which lets us think of an element of $\h_3(\O)$ as a pair
\[    (X,M) \in  \h_3(\C) \oplus \M_3(\C)  \]
Any element $(g,h) \in \SU(3) \times \SU(3)$ acts on such pairs as follows:
\[    (g,h) \maps (X,M) \mapsto (gXg^\dagger, hMg^\dagger ) \]
Note that $(e^{2\pi i/3}, e^{2\pi i/3}) \in \SU(3) \times \SU(3)$ acts trivially, so we get a representation of $(\SU(3) \times \SU(3))/\Z_3$ on $\h_3(\O)$.  

One can see that this action of $(\SU(3) \times \SU(3))/\Z_3$ on $\h_3(\O)$ preserves:
\begin{itemize}
\item the subalgebra $B \cong \h_3(\C)$ (by construction)
\item the Jordan product on $\h_3(\O)$ (by a straightforward calculation).
\end{itemize}
Using Borel--de Siebenthal theory \cite{BS,Dynkin,Yokota}, one can show that this copy of $(\SU(3) \times \SU(3))/\Z_3$ is a maximal connected subgroup of $\F_4$.  Thus, this subgroup is the \emph{largest} connected subgroup of $\F_4$ that preserves $B$:
\[   \Stab(B)_0 \cong (\SU(3) \times \SU(3))/\Z_3 .\]

The two $\SU(3)$'s in  $(\SU(3) \times \SU(3))/\Z_3$ act very differently on $\h_3(\O)$.  The first $\SU(3)$ acts to \emph{mix up} the matrix entries:
\[    g \maps (X,M) \mapsto (gXg^\dagger, Mg^\dagger ) \]
and only its subgroup $\S(\U(2) \times \U(1))$ preserves $A$.  This subgroup becomes the electroweak gauge group:
\[   \S(\U(2) \times \U(1)) \cong \big(\SU(2) \times \U(1)\big)/\Z_2 \]
The second $\SU(3)$ becomes the strong force gauge group: it acts \emph{separately} on the matrix entries $x,y,z$ in  
\[    \left( \begin{array}{ccc}  
                         \alpha  &  z  & y^*    \\  
                         z^*       & \beta & x      \\ 
                         y       & x^*   & \gamma   
\end{array} \right) \in \h_3(\O) 
\]
as octonion automorphisms that preserve $\C \subset \O$.  Thus, all of this copy of $\SU(3)$ preserves $A$.  Therefore
\[   \Stab(A) \cap \Stab(B)_0  \; \cong \; \big(\S(\U(2) \times \U(1)) \times \SU(3)\big) /\Z_3  \]
But
\[     
  \big(\S(\U(2) \times \U(1)) \times \SU(3)\big) /\Z_3  \;\cong\; \S(\U(2) \times \U(3))  
\]
We thus conclude
\[   \Stab(A) \cap \Stab(B)_0 \cong \; \big(\S(\U(2) \times \U(1)) \times \SU(3)\big) /\Z_3   \qedhere \]
\end{proof}

Todorov and Dubois-Violette \cite{TDV} emphasized that $\Stab(A) \cong \Spin(9)$ and $\Stab(B)_0 \cong (\SU(3) \times \SU(3))/\Z_3$ are maximal connected closed subgroups of $\F_4$.  However, in Remark 2 of Section 2.12, Yokota \cite{Yokota} notes that $\Stab(B)_0$ is not a maximal closed subgroup of $\F_4$: it is contained in a larger closed subgroup with at least two connected components.  One can check that both these components lie in $\Stab(B)$, so $\Stab(B)$ is not connected.

This is why taking the identity component of $\Stab(B)$ is necessary in Theorems \ref{thm:1} and \ref{thm:2}.  For example, $\Stab(X) \cap \Stab(B)$ is strictly larger than $\Stab(X) \cap \Stab(B)_0 \cong \S(\U(2) \times \U(3))$.  It includes not only elements of $\F_4$ that act on $\h_3(\C)$ via $X \mapsto UXU^{-1}$ where $U$ is a unitary transformation of $\C^3$, but  also, in another connected component, elements that act via $X \mapsto UXU^{-1}$ where $U$ is antiunitary.  The Standard Model does have an antiunitary symmetry, namely CPT symmetry, but this is not usually included in the gauge group.

\section{Transitivity of $\F_4$ on $\h_3(\C)$-subalgebras}
\label{sec:transitivity}

To prove Theorem \ref{thm:2} starting from the special case in Lemma \ref{lem:thm:2_special_case}, we shall show that $\F_4$ acts to map any other case to this case.   As a step toward this result, here we show that $\F_4$ acts transitively on the set of Jordan subalgebras of $\h_3(\O)$ that are isomorphic to $\h_3(\C)$.  

For this we use the theory of Jordan frames \cite[Chap.\ IV]{FarautKoranyi}.  Let $\K$ be a normed division algebra and let $n \ge 1$.  Let $\h_n(\K)$ be the space of $n \times n$ self-adjoint matrices with entries in $\K$, equipped with the product
\[     a \circ b = \frac{1}{2}(a b + b a) .\]
Then $\h_n(\K)$ is a Jordan algebra if $\K \cong \R, \C$ or $\H$, or if $\K \cong \O$ and $n \le 3$.   Assume henceforth that we are in one of these cases.   

An \define{idempotent} in $\h_n(\K)$ is an element $e$ with $e \circ e = e$.   Two idempotents $e, f$ are said to be \define{orthogonal} if $e \circ f = 0$, or equivalently $\langle e, f \rangle = 0$ with respect to the inner product
\[      \langle x, y \rangle = \tr(x \circ y) \]
on $\h_n(\K)$.  An idempotent is \define{minimal} if it is not a sum of pairwise orthogonal idempotents in a nontrivial way, or equivalently, its trace equals $1$.  A set of minimal idempotents that are pairwise orthogonal is called a \define{Jordan frame}.  Every Jordan frame of $\h_n(\K)$ consists of $n$ elements, as illustrated by the \define{standard} Jordan frame of $\h_3(\K)$:
\[ 
e_1 =  
\left( \begin{array}{ccc}
1 & 0 & 0 \\
0 & 0 & 0 \\
0 & 0 & 0 
\end{array}
\right) , \quad
e_2 =  
\left( \begin{array}{ccc}
0 & 0 & 0 \\
0 & 1 & 0 \\
0 & 0 & 0 
\end{array}
\right) , \quad
e_3 =  
\left( \begin{array}{ccc}
0 & 0 & 0 \\
0 & 0 & 0 \\
0 & 0 & 1 
\end{array}
\right) .
\]

We also use the theory of Peirce decompositions \cite[Chap.\ IV]{FarautKoranyi}.  For every idempotent $e\in\h_3(\K)$, the operation of multiplying by $e$ has at most three eigenvalues: $0,1/2$ and $1$.  Indeed $\h_n(\K)$ has a \define{Peirce decomposition}, an orthogonal direct sum decomposition
\[\h_3(\K)=E_0(e)\oplus E_{1/2}(e)\oplus E_1(e)\]
into the eigenspaces of multiplication by $e$:
\begin{equation}
\label{eq:peirce}
E_\lambda(e)=\{X\in\h_n(\K)\,|\,e \circ X=\lambda X\}.\end{equation}

We can use the Peirce decompositions of the idempotents $e_1,e_2,e_3$ to decompose $\h_3(\K)$ into 6 orthogonal subspaces.  First, we define linear injections $\xi_{ij}: \K\to\h_3(\K)$ by
\begin{align*}
\xi_{12}(x)&=\begin{pmatrix}0&x&0\\x^\ast&0&0\\0&0&0\end{pmatrix},&
\xi_{23}(x)&=\begin{pmatrix}0&0&0\\0&0&y\\0&y^\ast&0\end{pmatrix},&
\xi_{31}(x)&=\begin{pmatrix}0&0&z^\ast\\0&0&0\\z&0&0\end{pmatrix}
\end{align*}
for $x\in\K$.  This gives an orthogonal sum decomposition
\begin{equation}
\label{eq:decomposition_into_6_summands}
\begin{array}{ccl} 
\h_3(\K)&=&\spann_\R\{e_1\} \oplus \spann_\R\{e_2\} \oplus \spann_\R\{e_3\} \\ [3 pt]
& & \oplus \,\, \xi_{12}(\K) \, \oplus \, \xi_{23}(\K) \, \oplus \,\xi_{31}(\K). 
\end{array}
\end{equation}
Then we can relate these 6 summands to the Peirce eigenspaces of the standard Jordan frame:
\begin{equation}
\begin{array}{rcl}
\spann_\R\{e_i\} &=& E_1(e_i),  \\ [3 pt]
\xi_{ij}(\K)&=&E_{1/2}(e_i)\cap E_{1/2}(e_j),\\  [3 pt]
\spann_\R\{e_i,e_j\}\oplus\xi_{ij}(\K)&=& E_0(e_k), \\
\end{array}
\label{eq:peirce_facts}
\end{equation}
where $(i,j,k)$ is a cyclic permutation of $(1,2,3)$.  These are easy to check by direct calculation.

We shall also use these standard facts about the action of $\F_4$ on Jordan frames:

\begin{lem}
\label{lem:Jordan_frames}
The group $\F_4$ acts transitively on the set of Jordan frames of $\h_3(\O)$.  The stabilizer of any Jordan frame of $\h_3(\O)$ is isomorphic to $\Spin(8)$.  For the standard Jordan frame, the action of $\Spin(8)$ on $\h_3(\K)$ preserves all 6 subspaces in Equation \eqref{eq:decomposition_into_6_summands}.  On $\xi_{12}(\O), \xi_{23}(\O)$ and $\xi_{31}(\O)$ we can take this action to be the vector, left-handed real spinor, and right-handed real spinor representation of $\Spin(8)$, respectively.

\end{lem}

\begin{proof}
For any simple Jordan algebra $A$, $\Aut(A)$ acts transitively on the set of Jordan frames \cite[Thm.\ IV.2.5]{FarautKoranyi}.   Jordan, Wigner and von Neumann showed that $\h_3(\O)$ is simple \cite{JVW}.   Yokota \cite[Thm.\ 2.7.1]{Yokota} presented a proof that the stabilizer of the standard Jordan frame is $\Spin(8)$, which can be taken to act on on $\xi_{12}(\O), \xi_{23}(\O)$ and $\xi_{31}(\O)$ via the vector, left-handed real spinor, and right-handed real spinor representations of $\Spin(8)$,  respectively.  By transitivity, the stabilizer of any Jordan frame is isomorphic to $\Spin(8)$.  
\end{proof}

Using this machinery, we can prove that $\F_4$ acts transitively on the set of Jordan subalgebras that are isomorphic to $\h_3(\C)$:

\begin{lem}
\label{lem:transitivity_for_subalgebras}
$\F_4$ acts transitively on the set of Jordan subalgebras $B \subset \h_3(\O)$ with $B \cong \h_3(\C)$.
\end{lem}

\begin{proof}
It suffices to show that for any Jordan subalgebra $B\subset\h_3(\O)$ isomorphic to $\h_3(\C)$ can be mapped, using some element of $\F_4$, to the standard copy of $\h_3(\C)$ given in Equation \eqref{eq:choice_of_AB}.

Since $B \cong \h_3(\C)$, it must contain some Jordan frame with three elements.  By virtue of Lemma \ref{lem:Jordan_frames}, we can apply some element $g \in \F_4$ to map this frame to the standard Jordan frame in $\h_3(\C)$, which is also the standard frame of $\h_3(\O)$.  Thus, without loss of generality we may assume $B$ contains $e_1, e_2, e_3 \in \h_3(\O)$.   

Next, note that the Peirce decomposition is compatible with passing to a subalgebra: that is, for any idempotent $e\in \h_3(\O)$ that happens to lie in $B$, we have
\[ B=(E_0(e)\cap B)\oplus(E_{1/2}(e)\cap B)\oplus(E_1(e)\cap B).\]
Combining Equations \eqref{eq:decomposition_into_6_summands} and \eqref{eq:peirce_facts}, we obtain
\begin{align*}
B=\spann_\R\{e_1,e_2,e_3\}&\oplus(E_{1/2}(e_1)\cap E_{1/2}(e_2)\cap B)\\
&\oplus(E_{1/2}(e_2)\cap E_{1/2}(e_3)\cap B)\\
&\oplus(E_{1/2}(e_3)\cap E_{1/2}(e_1)\cap B).
\end{align*}
By Equation \eqref{eq:decomposition_into_6_summands}, $E_{1/2}(e_i)\cap E_{1/2}(e_j)$ is the image of $\xi_{ij}$.   We define 
\[  V_{ij} =\xi_{ij}^{-1}(E_{1/2}(e_i)\cap E_{1/2}(e_j)\cap B) \subset \O .\]
A matrix calculation shows that
\begin{equation}
\xi_{ij}(x)\circ\xi_{jk}(y)=\xi_{ki}(y^\ast x^\ast)
\label{eq:xi_mult}
\end{equation}
for all cyclic permutations $(i,j,k)$ of $(1,2,3)$. Hence if $x\in V_{ij}$ and $y\in V_{jk}$, then $y^\ast x^\ast=(xy)^\ast\in V_{ki}$, showing 
\begin{equation}
\label{eq:decomposition}
\begin{array}{ccl}
B &=&\left\{\begin{pmatrix}
    \alpha_1&x&z^\ast\\
    x^\ast&\alpha_2&y\\
    z&y^\ast&\alpha_3\end{pmatrix} \; \colon \; \alpha_i\in\R,\ x\in V_{12},\ y\in V_{23}, z\in V_{31}\right\}\\ \\
&=& \spann_\R\{e_1,e_2,e_3\}\oplus\xi_{12}(V_{12})\oplus\xi_{23}(V_{23})\oplus\xi_{31}(V_{31}),
\end{array}
\end{equation}
and 
\begin{align}
V_{12}\cdot V_{23}&\subseteq V_{31}^\ast,&
V_{23}\cdot V_{31}&\subseteq V_{12}^\ast,&
V_{31}\cdot V_{12}&\subseteq V_{23}^\ast.
\label{eq:mult_rel1}
\end{align}

By our assumptions, there is an isomorphism $B \cong \h_3(\C)$ preserving the elements $e_1,e_2,e_3 \in \h_3(\O)$.  Thus $E_{1/2}(e_i)\cap E_{1/2}(e_j)\cap B$ is isomorphic to $\xi_{ij}(\C) = E_{1/2}(e_i)\cap E_{1/2}(e_j)\cap\h_3(\C)$, which has dimension 2. Since $\xi_{ij}$ is injective, it follows that
\[
\dim_\R(V_{12}) = \dim_\R(V_{23}) = \dim_\R(V_{31}) = 2 .
\]

From Lemma~\ref{lem:Jordan_frames} we know that the subgroup of $\F_4$ stabilizing the standard Jordan frame is $\Spin(8)$, which preserves the three spaces $\xi_{ij}(\O)$, and can be taken to act on on $\xi_{12}(\O), \xi_{23}(\O)$ and $\xi_{31}(\O)$ via the vector, left-handed real spinor, and right-handed real spinor representations of $\Spin(8)$, respectively. 

Each map $\xi_{ij} \colon \O\to\h_3(\O)$ is $\Spin(8)$-equivariant for the corresponding representation of $\Spin(8)$ on $\O$.  Clearly the vector representation of $\Spin(8)$ acts transitively on the Grassmannian of 2-planes in $\O$.  Thus, we can act on $B$ with an element $g \in \Spin(8) \subset \F_4$ so that the 2-dimensional space $V_{12}$ gets mapped to $\C\subset\O$.  

This allows us to assume without loss of generality that $V_{12}=\C$, which we do henceforth.   Choosing some unit vector $a \in V_{23}$, we then have
\begin{equation}
\label{eq:B_decomposition}
B=\spann\{e_1,e_2,e_3\}\oplus\xi_{12}(\C)\oplus\xi_{23}(\C\cdot a)\oplus\xi_{31}(a^\ast\cdot\C) 
\end{equation}

To find an element of $\F_4$ mapping $B$ to the standard copy of $\h_3(\C)$ in $\h_3(\O)$, it then suffices to find $g \in \Spin(8)$ whose action in the vector representation of $\Spin(8)$ preserves $\C \subset \O$, while its action in the left-handed spinor representation maps $a$ to $1 \in \O$.  Such an element $g$ exists because the stabilizer of $\C \subset \O$ in the vector representation is $(\Spin(6) \times \Spin(2))/\Z_2 \subset \Spin(8)$, and the $\Spin(6)$ subgroup acts transitively on the unit sphere of $\O$ in the left-handed spinor representation of $\Spin(8)$.  This transitivity follows from the fact that $\Spin(6) \cong \SU(4)$, with the two spinor representations of $\Spin(8)$ on $\O \cong \C^4$ restricting to representations that are isomorphic to the tautologous representation of $\SU(4)$ on $\C^4$ and its dual \cite{Adams,Harvey}.
\end{proof}

\section{Proofs of theorems}
\label{sec:proofs}

Next we prove that $\F_4$ acts transitively on pairs of Jordan subalgebras $A,B \subset \h_3(\O)$ with $A \cong \h_2(\O)$ and $B \cong \h_3(\C)$ such that $A \cap B \cong \h_2(\C)$.  This quickly leads to a proof of Theorem \ref{thm:2}, and with more reflection a proof of Theorem \ref{thm:1}.  

To get started, we need some facts about Jordan subalgebras of $\h_3(\O)$ that are isomorphic to $\h_2(\C)$ or $\h_2(\O)$.  Even though we are only interested in $\C$ and $\O$, the following four lemmas are just as easy to prove in greater generality.  Thus, let $\K$ and $\L$ denote any of the normed division algebras $\R$, $\C$, $\H$ or $\O$.  As it turns out, the hypotheses of Lemmas \ref{lem:unitisline} and \ref{lem:uniqueness} can only hold when $\dim(\L) \le \dim(\K)$, but all the lemmas are true in general.

\begin{lem}
\label{lem:E1subalgebra}
Let $\ell\in\h_3(\K)$ be an idempotent with $\tr(\ell)=2$. Then the Peirce subspace $E_1(\ell) = \{X\in\h_3(\K)\,|\,\ell \circ X= X\}$ is a subalgebra of $\h_3(\K)$ isomorphic to $\h_2(\K)$.
\end{lem}
\begin{proof}
If we set $p=1-\ell$, then $p$ is an idempotent of trace $1$. The automorphism group of a simple Jordan algebra acts transitively on those \cite[Cor.~IV.2.4]{FarautKoranyi}, so without loss of generality we can assume that $p=e_3$, that is $\ell=e_1+e_2$. Then
\[E_1(\ell) = \left\{\begin{pmatrix}
    \alpha_1 & x &0 \\
    x^\ast & \alpha_2 & 0 \\
    0 & 0 & 0 \end{pmatrix} \; \colon \; \alpha_i\in\R, x\in \K \right\} \cong \h_2(\K).\]
\end{proof}

\begin{lem}
\label{lem:unitisline}
Let $A\subset\h_3(\K)$ be a subalgebra isomorphic to 
$\h_2(\L)$. Then the unit of $A$, regarded as an element $\ell \in \h_3(\K)$, is an idempotent with $\tr(\ell)=2$.
\end{lem}
\begin{proof}
Let $f\maps\h_2(\L)\to\h_3(\K)$ be an injective algebra homomorphism with image $A$; then $\ell=f(1)$.   Since $1$ is an idempotent, so is $\ell$, so we only need to show that the trace of $\ell$ is $2$.

In $\h_2(\L)$ let $p_1=(\begin{smallmatrix}1&0\\0&0\end{smallmatrix})$, $p_2=(\begin{smallmatrix}0&0\\0&1\end{smallmatrix})$ and $u=(\begin{smallmatrix}0&1\\1&0\end{smallmatrix})$. These satisfy
\[
p_i^2 = 1, \quad
u^2 = p_1+p_2  =1, \quad
p_i\circ u =\tfrac{1}{2} u.
\]
The elements $q_1=f(p_1)$, $q_2=f(p_2)$ and $\ell=q_1+q_2$ are themselves nonzero idempotents, which means they all have trace $1$, $2$ or $3$.  Thus, if $\tr(q_1)=\tr(q_2)$, the only possibility is that both equal $1$.
 
By the identity $\tr((a \circ b) \circ c) = \tr(a \circ (b \circ c))$ \cite[Prop. II.4.3]{FarautKoranyi}, we find
\begin{align*}
\tr(q_i)&=\tr(f(p_i\circ 1))=\tr(f(p_i\circ u^2)) =
\tr (q_i\circ f(u)^2)=\tr((q_i\circ f(u))\circ f(u))\\
&=\tr(f(p_i\circ u)\circ f(u))=\tfrac12\tr(f(u)^2),
\end{align*}
which is independent of $i=1,2$. Thus $\tr(q_1)=\tr(q_2)$, hence they both equal $1$, so $\tr(\ell)=2$.
\end{proof}

\begin{lem}
\label{lem:subalgebraE1}
Every subalgebra of $\h_3(\K)$ isomorphic to $\h_2(\K)$ is of the form $E_1(\ell)$ for a unique idempotent $\ell$ with $\tr(\ell)=2$.
\end{lem}
\begin{proof}
Let $A\subset\h_3(\K)$ be such a subalgebra, and let $\ell$ be the unit of $A$ regarded as an element of $\h_3(\K)$. Then $\tr(\ell)=2$ by Lemma~\ref{lem:unitisline}, so Lemma~\ref{lem:E1subalgebra} implies that $E_1(\ell)$ is a subalgebra of $\h_3(\K)$ isomorphic to $\h_2(\K)$. Clearly $A\subseteq E_1(\ell)$, and since both sides are isomorphic to $\h_2(\K)$, they must be equal.

To show that $\ell$ is the unique idempotent such that $A=E_1(\ell)$, assume that $A=E_1(\ell')$ for another idempotent $\ell'\in\h_3(\K)$. Since $\ell\in E_1(\ell)=E_1(\ell')$ and $\ell'\in E_1(\ell')=E_1(\ell)$ it follows that
\[\ell=\ell\circ\ell'=\ell'. \qedhere \]
\end{proof}

\begin{lem}
\label{lem:uniqueness}
Every Jordan subalgebra $B \subseteq \h_3(\K)$ isomorphic to $\h_2(\L)$ is contained in a unique Jordan subalgebra isomorphic to $\h_2(\K)$.   If $\ell$ is the unit of $B$ regarded as an element of $\h_3(\K)$, this subalgebra is 
\[  E_1(\ell) = \{ a \in \h_3(\K) : \; \ell \circ a = 1 \}.\]
\end{lem}

\begin{proof}
Let $\ell$ be the unit of $B$, regarded as an element of $\h_3(\K)$. Then Lemma~\ref{lem:unitisline} shows that $\ell$ is an idempotent with trace 2, so Lemma~\ref{lem:E1subalgebra} implies that $E_1(\ell)$ is a subalgebra isomorphic to $\h_2(\K)$. Clearly $B\subseteq E_1(\ell)$.

For uniqueness, note that by Lemma~\ref{lem:subalgebraE1} every $\h_2(\K)$-subalgebra of $\h_3(\K)$ is of the form $E_1(\ell')$ for some trace $2$ idempotent $\ell'$. If $B\subseteq E_1(\ell')$, then in particular $\ell\in E_1(\ell')$, and the same argument as in the proof of Lemma~\ref{lem:subalgebraE1} shows that $\ell=\ell'$.
\end{proof}

With these facts in hand, we prove:

\begin{lem}
\label{lem:transitivity_for_pairs}
$\F_4$ acts transitively on pairs of Jordan subalgebras $A,B \subset \h_3(\O)$ with $A \cong \h_2(\O)$, $B \cong \h_3(\C)$ and $A \cap B \cong \h_2(\C)$.
\end{lem}

\begin{proof} 
Let $(A,B), (A',B')$ be pairs as above.  It suffices to find $g \in \F_4$ with $g A = A'$ and $g B = B'$.  But in Lemma \ref{lem:transitivity_for_subalgebras}, we saw that $\F_4$ acts transitively on the set of Jordan subalgebras $B \subset \h_3(\O)$  with  $B \cong \h_3(\C)$.   Given this we can assume without loss of generality that $B= B'$ is the standard copy of $\h_3(\C)$ given as in Equation \eqref{eq:choice_of_AB}.  Henceforth we call this copy simply $\h_3(\C)$.  It then suffices to find $g \in \Stab(\h_3(\C))$ with $g A = A'$.

By assumption, $A \cap \h_3(\C)$ and $A' \cap \h_3(\C)$ are Jordan subalgebras isomorphic to $\h_2(\C)$.    $\SU(3)$ acts on $\h_3(\C)$ by
\[          g \maps X \to g X g^\dagger  \]
and it acts transitively on trace $2$ idempotents in $\h_3(\C)$, which correspond to two-dimensional subspaces of $\C^3$. Therefore, by Lemma~\ref{lem:subalgebraE1}, it also acts transitively on the set of subalgebras of $\h_3(\C)$ isomorphic to $\h_2(\C)$. Thus, we can find $g_0 \in \SU(3)$ moving $A \cap \h_3(\C) $ to $A' \cap \h_3(\C)$. 

In the proof of Lemma \ref{lem:thm:2_special_case} we saw $\Stab_0(\h_3(\C)) = (\SU(3) \times \SU(3))/\Z_3$, and  $g = [g_0,1] \in \Stab_0(\h_3(\C))$ acts on $\h_3(\C)$ just as $g_0$ does.   Thus $g$ stabilizes $\h_3(\C)$ and moves $A \cap \h_3(\C)$ to $A' \cap \h_3(\C)$.   

To conclude, we just need to show $gA = A'$.   For this we use Lemma \ref{lem:uniqueness}, which says that any Jordan subalgebra of $\h_3(\O)$ isomorphic to $\h_2(\C)$ is contained in a unique subalgebra isomorphic to $\h_2(\O)$.  $gA$ is a Jordan algebra isomorphic to $\h_2(\O)$ that contains  $g(A \cap \h_3(\C)) = A' \cap \h_3(\C)  \cong \h_2(\C)$, but $A'$ is another subalgebra with this property, so we must have $g A = A'$.   
\end{proof}

The proof of Theorem \ref{thm:2} is now immediate.

\addtocounter{thm}{-10}%
\begin{thm}[\textup{\textbf{restated}}]
Suppose $A,B$ are Jordan subalgebras of $\h_3(\O)$ such that
\[  A \cong \h_2(\O), \;\; B \cong \h_3(\C), \;\; A \cap B \cong \h_2(\C) . \]
Then 
\[     \Stab(A) \cap \Stab(B)_0 \cong \S(\U(2) \times \U(3)) .\]
\end{thm}

\addtocounter{thm}{5}%
\begin{proof}
Lemma \ref{lem:thm:2_special_case} says that this theorem holds in the special case where $A$ and $B$ are given as in Equation \eqref{eq:choice_of_AB}.  It thus suffices to show that we can transform any case to this special case using the action of $\F_4$.  This is the content of Lemma \ref{lem:transitivity_for_pairs}.  
\end{proof}

The proof of Theorem \ref{thm:2} contains within it the seeds of this result, which has simpler hypotheses:

\addtocounter{thm}{-7}%
\begin{thm}[\textup{\textbf{restated}}]
Suppose $X,B$ are Jordan subalgebras of $\h_3(\O)$ such that
\[  X \cong \h_2(\C), \;\; B \cong \h_3(\C), \;\; X \subset B . \]
Then 
\[     \Stab(X) \cap \Stab(B)_0 \cong \S(\U(2) \times \U(3)). \]
\end{thm}  
\addtocounter{thm}{6}%

\begin{proof}
Suppose $X$ and $B$ have the stated properties.  By Lemma \ref{lem:uniqueness}, $X \cong \h_2(\C)$ is contained in a unique Jordan subalgebra $A \cong \h_2(\O)$ of $\h_3(\O)$, namely $A = E_1(\ell)$ where $\ell\in\h_3(\O)$ is the unit of $X$.

Let $f\maps\h_3(\C)\to B$ be an algebra isomorphism. Then $f^{-1}(\ell)$ is the unit of the subalgebra $f^{-1}(X)\cong\h_2(\C)$, so it follows from Lemma~\ref{lem:unitisline} that $\tr(f^{-1}(\ell))=2$. Therefore, by Lemma~\ref{lem:E1subalgebra}, $E_1(f^{-1}(\ell))\subseteq\h_3(\C)$ is a subalgebra isomorphic to $\h_2(\C)$. But so is $f^{-1}(X)\subseteq E_1(f^{-1}(\ell))$, so they are equal. Thus
\[X=f(E_1(f^{-1}(\ell)))=E_1(\ell)\cap B=A\cap B.\]
It follows that $A$ and $B$ obey the hypotheses of Theorem \ref{thm:2}, so 
\[   \Stab(A) \cap \Stab(B)_0 \cong \S(\U(2) \times \U(3)). \]
Since $X$ uniquely determines $A$ by Lemma \ref{lem:uniqueness}, we have
\[   \Stab(X) \cap \Stab(B)_0 \subseteq \Stab(A) \cap \Stab(B)_0 .\]
On the other hand, since $A$ and $B$ determine $X = A \cap B$, we have
\[   \Stab(X) \cap \Stab(B)_0 \supseteq \Stab(A) \cap \Stab(B)_0 .\]
Thus 
\[  \Stab(X) \cap \Stab(B)_0  = \Stab(A) \cap \Stab(B)_0 \cong \S(\U(2) \times \U(3)). \qedhere \]
\end{proof}

\subsection*{Acknowledgements}

We thank David Madore for help in coming up with Lemmas \ref{lem:uniqueness} and \ref{lem:transitivity_for_pairs}.

\end{document}